\documentclass[10pt, conference, letterpaper]{IEEEtran}
\pdfoutput=0
\usepackage{amsmath}
\usepackage{amsthm}
\usepackage{amssymb}
\usepackage[textsize=small, color=red]{todonotes}
\usepackage{fullpage}
\usepackage{graphicx}
\usepackage{tikz}
\usepackage[linesnumbered,ruled,vlined]{algorithm2e}
\usepackage{enumitem}

\DeclareMathOperator*{\argmax}{arg\,max}
\DeclareMathOperator*{\argmin}{arg\,min}

\usepackage{amsbsy}
\newcommand{\pot}{\operatorname{Pot}}
\newcommand{\supp}{\operatorname{supp}}
\newcommand{\ft}{\widetilde{f}}

\newcommand{\vx}{\vec{x}}
\newcommand{\vxe}{\vec{x}^*}
\newcommand{\vl}{\vec{\ell}}
\newcommand{\va}{\vec{a}}
\newcommand{\vb}{\vec{b}}
\newcommand{\vy}{\vec{y}}
\newcommand{\vz}{\vec{z}}
\newcommand{\vs}{\vec{s}}
\newcommand{\ve}{\vec{e}}
\newcommand{\vm}{\vec{m}}
\newcommand{\vM}{\vec{M}}
\newcommand{\gm}{\gamma_m}
\newcommand{\vlk}[1]{\vl^{(#1)}}
\newcommand{\ellk}[1]{\ell^{(#1)}}
\newcommand{\mk}[1]{M^{(#1)}}
\newcommand\bigo[1]{\mathcal{O}\!\left( #1\right) }
\newcommand\bigom[1]{\Omega\!\left( #1\right)}
\newcommand\mps[2]{\left<\,#1\,|\,#2\,\right>}
\newcommand\floor[1]{\left\lfloor #1 \right\rfloor}
\newcommand\ceil[1]{\left\lceil #1 \right\rceil}

\newcounter{fabrice}

\newcounter{corinne}

\newcommand{\normu}[1]{{\left\lVert#1\right\rVert}_\infty}
\newtheorem{theorem}{Theorem}
\newtheorem{lemma}[theorem]{Lemma}
\newtheorem{definition}[theorem]{Definition}
\newtheorem{proposition}[theorem]{Proposition}

\newif\iffullarticle
\fullarticlefalse

\begin{document}
\title{The Social Medium Selection Game }
\date{}

\author{
\IEEEauthorblockN{Fabrice Lebeau\IEEEauthorrefmark{1}\IEEEauthorrefmark{2}, Corinne Touati\IEEEauthorrefmark{1}\IEEEauthorrefmark{3}, Eitan Altman\IEEEauthorrefmark{1} and Nof Abuzainab\IEEEauthorrefmark{1}}
\IEEEauthorblockA{\IEEEauthorrefmark{1}Inria \quad
\IEEEauthorrefmark{3} CNRS, LIG, Univ. Grenoble Alpes
\quad \IEEEauthorrefmark{2}ENS Lyon
\\
Email: fabrice.lebeau@ens-lyon.fr, \{corinne.touati, eitan.altman, nof.abuzainab\}@inria.fr}}
\maketitle

\IEEEpeerreviewmaketitle

\begin{abstract}
We consider in this paper competition of content creators in routing their content through various media. The routing decisions may correspond to the selection of a social network (e.g. twitter versus facebook or linkedin) or of a group within a given social network. The utility for a player to send its content to some medium is given as the difference between the dissemination utility at this medium and some transmission cost. We model this game as a congestion game and compute the pure potential of the game. In contrast to the continuous case, we show that there may be various equilibria. We show that the potential is M-concave which allows us to characterize the equilibria and to propose an algorithm for computing it. We then give a learning mechanism which allow us to give an efficient algorithm to determine an equilibrium. We finally determine the asymptotic form of the equilibrium and discuss the implications on the social medium selection problem.
\end{abstract}

\section{Introduction}
Social networks involve many actors who compete over many resources. This gives rise to competitions at different levels which need to be taken into account in order to explain and predict the system behavior. In this paper, we focus on competition  of individual content creators over media. A content creator has to decide which one of several media to use. The media choice may correspond to a social network that will be used for sending (and disseminating) some content. For instance, the decision can consist in choosing between twitter and facebook, or in deciding to which of several facebook groups to send the content. 

The game we study in this paper is atomic and non-splitable. We consider a decision maker (or a player) to be a single content instead. This regime can well approximate decision making where a content creator, say a blogger, occasionally sends content. Here, occasionally implies that the time intervals between generation of consecutive contents by the blogger is large enough so that the states of the system at the different times of creation of content are independent one from another. This regime is interesting not only because it is characteristic of systems with many sources of contents, but also and foremost, it turns out that it precisely characterizes bloggers that have more popularity and influence. This was established experimentally in \cite{mckl} which analyses the role of intermediate actors in dissemination of content. 

A similar game as the one in this paper was already studied in \cite{EitanDynamicModel} and \cite{RAH}, but there the players
control rates of creation of contents and/or decide how to split the rates. The resulting games are simpler than
ours as they possess a single equilibrium. The model studied in this paper brings many novelties both in the system behavior as well as in the tools used to study it. The difficulty in studying the game in our atomic non-splitable game framework is due to the integrity constraint on the players as they cannot split their content between several media. This implies that the action space is discrete and thus non convex which
may result in problems in the existence and/or uniqueness of the equilibrium.

\paragraph*{Related work}
Game theoretic models for competition have been proposed in a growing number of references.
 The authors of \cite{hmv} focus on the competition over budget of attention of content 
consumers and the impact of this competition on the dynamic popularity of the content.
A game model related to
intermediate actors that participate in the spreading of news
is considered in \cite{mckl}. The authors study how to choose the
type  and amount of content to send so as to be influential. In \cite{LPT}, the authors study
competition over space among content creators. The space may represent a slot (say the top one)
in  a timeline, and a content that arrives occupies the space pushing out the one that is already
there. It then stays visible there till the 
next arrival of content that pushes it away.
The authors study the timing game: when should a content be sent to the timeline so as
to maximize the expected time it remain visible. The authors of \cite{MP} study a dynamic competition model
over visibility in which the rate at which creator of contents send their traffic is controlled.

Game theory has been used not only to model competition in social networks but also to design algorithms for the analysis of social networks \cite{NN11}: this includes community detection \cite{NN12}, discovery of influential nodes \cite{NN11b} and more.

\paragraph*{Contributions of our paper}
Our first contribution is to make the observation that the game complies with the definition of 
congestion game introduced by Rosenthal \cite{rsntl}. This allows us to show that our game is a potential game, for which Nash equilibria exist. This further implies that algorithms based on best response converge to an equilibrium. We show that
surprisingly, although the potential approximates a strictly concave function (in the continuous space), there may exist many equilibria. This is quite a new
phenomenon in networking games, and it is due to the non convexity of the action space (due to the
non splitable assumption). In order to have uniqueness, a new concept of integer concavity have
to be used. We rely on the theory of M-concavity~\cite{MurotaMConvex} which allows us to establish the structure of the set of equilibria for this problem. We propose a learning algorithm that converges to an equilibrium and is more efficient than the best response algorithm. We finally study the asymptotic behavior of the system as the number of players grows, both in terms of characterization of the equilibria and in price of anarchy.


\section{Model and Notations}

We consider a set $\mathbb{K} = \{1,...,K\}$ of seeds (content producers,
bloggers, etc) that aim to publish their content in social media. We focus on the problem where each seed needs to publish in some social medium $j \in \mathbb{J} =  \{1,...,J\}$. 
The strategy $s_k$ of seed $k$ is the social medium it selects for
disseminating its content. 
Define the load $\ell_j$ on social medium $j$ as the
number of competing seeds that send their content to medium $j$.
It can be written as
$ \ell_j = \sum_k \delta(s_k, j)$ ($\delta$ is the Kronecker symbol, i.e. $\delta(a,b)$ equals $1$ if $a=b$ and $0$ otherwise). 
Assume that social medium $j$ has $N_j > 0$ subscribers who are interested
in content shipped to that medium. 

The utility of a player (seed) is given as the
difference between a \emph{dissemination utility} and a \emph{dissemination cost}. 
The former, that is, the value of
disseminating a content at the $j$th social network is (i) proportional
to $N_j$ and (ii) inversely proportional to $\ell_j$.
Further, each seed pays a constant dissemination cost $\gamma_j$ for publishing on 
social medium $j$.

This structure of utility is very common. In the networking community,
we find it in resource sharing of link capacity for flow control
problems. It is also associated with the so-called
Kelly mechanism (see~\cite{MA} for a similar utility in cloud computing),
and models tracing back to the Tullock rent-seeking
 problem \cite{tullok}.
In the social medium context, this utility naturally arises if
different seeds create
similar content (say news) and thus a subscriber is not interested
in receiving more than one content. This implies the structure of point (ii).
The use of this type of utility in competition for resources in social 
networks can be found in \cite{EitanDynamicModel}.


Hence, the utility of seed $k$ is given by
\begin{equation}
u_k(\vec{s}) = \frac{ N_{s_k}} { \ell_{s_k} } - \gamma_{s_k},
\label{eq:utility}
\end{equation}
where $\vec{s}$ is the vector of strategies of the seeds $\vec{s} = (s_k)_{k \in \mathbb{K} }$

For notational convenience, in the following we will denote by $S = \mathbb{J}^{\mathbb{K}}$ the set of strategy profiles, and by $\Gamma=(K,(N_j,\gamma_j)_{j\in \mathbb{J}})$ the game setting with $K$ players (the seeds), the set $\mathbb{J}$ of social media of parameters $((N_j,\gamma_j)_{j\in \mathbb{J}})$ and the utility $(u_k)_{k \in \mathbb{K}}$ defined in Eq.~\eqref{eq:utility}.


Since the utility of each player (the seeds) only depends on the number of users choosing the same action (i.e. the same social media), the game is 
equivalent to a congestion game in the sense of Rosenthal \cite{rsntl}, where the resources are the social media. It therefore is a potential game~\cite{MondererShapley}, that is to say that there exists a function $\pot$ such that:
$$
\begin{array}{l}
\forall k \in \mathbb{K}, \forall \vec{s} \in S, \forall s_k, \\
\quad u_k(s_1,...,s_{k-1},s_k',s_{k+1},...,s_{K}) - u_k(\vec{s}) = \\
\qquad \pot (s_1,...,s_{k-1},s_k',s_{k+1},...,s_{K}) - \pot (\vec{s}).
\end{array}
$$

Let us introduce the Harmonic number:
$H_n = \sum_{j=1}^{n} 1/j$ if $n \geq 1$ and $H_n = 0$ if $n \leq 0$. 
Then, one can readily check that a suitable potential of the game is:
$$ \pot (\vs) = \sum_{j} \left( N_j H_{\ell_j} - \gamma_j \ell_j \right). $$

Let $\vl$ be the vector of loads induced by $\vs$ and $\mathcal{D}$ the set of possible vector loads: $\mathcal{D}=\left\{ (\ell_1,...,\ell_{J}) \in \mathbb{N}^{J}\ |\ \sum_{j=1}^{J} \ell_j = K \right\}$. By using the equivalence with
the congestion game of \cite{rsntl} we get
the potential as a function of 
the loads $\vl \in \mathcal{D}$ of the social media:
$$ \pot (\vl) = \sum_{j} \left( N_j H_{\ell_j} - \gamma_j \ell_j \right). $$


Let us finally introduce the following notations that will come in handy in the rest of the paper:
\begin{itemize}
\item For $\vx \in \mathbb{R}^n$, $\supp^+(\vx) = \{j\ |\ x_j > 0\}$ ,
\item $\mps{.}{.}$ is the euclidian scalar product over $\mathbb{R}^{J}$: $\mps{\vx}{\vy} = \sum_j x_j y_j$,
\item $\normu{.}$ is the uniform norm: $\normu{\vx} = \max_j |x_j|$,
\item $(\ve_j)_j$ is the Euclidean base of $\mathbb{R}^J$.
\end{itemize}

\section{Discrete Potential Analysis}

\subsection{Nash equilibria}
Potential games have received a lot of attention in the past years as they draw a natural bridge between the theory of games and optimization.
Indeed, the definition of a potential implies the following. 
A strategy profile $\vs$ is a Nash equilibrium iff it is a maximizer of the potential, that is, 
\begin{equation} \label{eq:eqS}
\begin{aligned}
 &\forall \vec{s}, \forall k \in \mathbb{K}\ \forall s_k' \in \mathbb{J},\\
			   \pot(&s_1,...,s_{k-1},s_k',s_{k+1},...,s_{K}) \leq \pot(\vs).
\end{aligned}
\end{equation}

Note that a change of strategy of a seeds from social media $i$ to social media $j$ amount in reducing the load $\ell_i$ of network $i$ and increase that of $j$ by one unit. Hence, in the space of load vectors, Equation~\eqref{eq:eqS} becomes:
\begin{equation} \label{eq:eqL}
\forall i,j: \ell_j > 0 \Rightarrow \pot(\vl+\ve_i-\ve_j) \leq \pot(\vl).
\end{equation}

For a given load vector $\vec{\ell}$, let $\mathcal{V}(\vec{\ell})$ be the set of possible load vectors obtained after the deviation of a single player:
$\phantom{\bigg\{} \hspace{-7pt}
\mathcal{V}(\vec{\ell}) = \left\{ \vl+\ve_i-\ve_j\ |\ i,j \in \mathbb{J} \text{ st } \ell_j>0\right\}$. Then, the Nash equilibria are all the strategy profiles $\vs$ for which the vector of loads is a local (in the sense of $\mathcal{V}$) maximum of $\pot:\mathcal{D} \rightarrow \mathbb{R}$. 

\subsection{M-concavity}

The potential is defined over a discrete set, and therefore classical convexity properties do not hold. In order to understand the structural and uniqueness properties of the Nash equilibrium, we study the properties of the potential function in terms of M-concavity\footnote{For more information about M-convexity, see~\cite[sec. 4.2]{MurotaMConvex}.}.

\begin{definition}
A function $f:\mathbb{Z}^J \rightarrow \mathbb{R}$ is M-concave if for all $\vx,\vy$ in $\mathcal{D}$ and for all $u \in \supp^+(\vx-\vy)$:
\begin{align*}
 &\exists v \in \supp^+(\vy-\vx),\\
 f(\vx) &+ f(\vy) \leq f(\vx-\ve_u+\ve_v) + f(\vy-\ve_v+\ve_u).
\end{align*}
\end{definition}

We have the fundamental property:
\begin{theorem}
The function $f:\mathbb{Z}^{J} \rightarrow \mathbb{R}$ defined by: 
$$ f(\vx) = \begin{cases}
                 \pot(\vx) &\text{ if } \vx \in \mathcal{D}, \\
                 -\infty &\text{ otherwise} 
                \end{cases} $$
is M-concave.  
\end{theorem}

\begin{proof}
Let $\vx,\vy \in \mathcal{D}$. First, if $\supp^+(\vx-\vy) = \emptyset$, then the property is trivially true.
Otherwise, assume that there is some $i$ in $\supp^+(\vx-\vy)$. If we had $\supp^+(\vy-\vx) = \emptyset$ then we would also have
$$ \sum_j y_j = y_i + \sum_{j \neq i} y_j < x_i + \sum_{j \neq i} x_j = \sum_j x_j = K. $$
This is absurd since $\vy \in \mathcal{D}$. Hence $\supp^+(\vy-\vx) \neq \emptyset$.

Then, let $u \in \supp^+(\vx-\vy)$ and $v \in \supp^+(\vy-\vx)$. Since $x_u > y_u \geq 0$ and $y_v > x_v \geq 0$ we have that $\vx-\ve_u+\ve_v$ and $\vy-\ve_v+\ve_u$ are in $\mathcal{D}$. Then:
\begin{align*}
 f(&\vx-\ve_u+\ve_v)+f(\vy-\ve_v+\ve_u)-f(\vx)-f(\vy) \\
&= N_v\underbrace{\left( \frac{1}{x_v+1}-\frac{1}{y_v}\right)}_{\geq 0} + N_u\underbrace{\left(\frac{1}{y_u+1}-\frac{1}{x_u}\right)}_{\geq 0}. \hspace{0.5em} \qedhere
\end{align*}
\end{proof}

Note that since we did not choose a particular $v$ in the preceding proof, we have actually shown a much stronger property, that is, that the inequality holds \emph{for any $v$}, which is decisive for the 
rest of the analysis:
\begin{equation}
 \label{eq:mconvex}
 \begin{aligned}
  &\forall \vx,\vy \in \mathcal{D},\ \\
&\forall u \in \supp^+(\vx-\vy),\ \mathbb{\boldsymbol{\pmb{\forall}}} v \in \supp^+(\vy-\vx),\: \\
  &f(\vx) + f(\vy) \leq f(\vx-\ve_u+\ve_v) + f(\vy-\ve_v+\ve_u).
 \end{aligned}
\end{equation}

\subsection{Properties of the Nash equilibria}
In this section, we show properties of the Nash equilibria of this game using the M-concavity of $f$.

\begin{theorem}
\label{thm:globMax}
 Let $\vs \in S$ and let $\vl$ be the vector of loads of $\vs$. Then: 
 \begin{center}
  $\vs$ is a Nash equilibrium for the game \\$\Leftrightarrow$ $\vl$ maximizes {\bf globally} the potential over $\mathcal{D}$. 
 \end{center}
\end{theorem}
\begin{proof}
 The sufficient condition is a direct consequence of Eq.~\eqref{eq:eqL}.
Conversely, assume that $\vs$ is a Nash equilibrium. Then, by~\eqref{eq:eqL}, we know that $\vl$ is a local maximum of $\pot$ (over $\mathcal{V}(\vl)$). Let $u,v \in \mathbb{J}$:\\
 If $\vl - \ve_u + \ve_v \notin \mathcal{D}$ then $f(\vl) > f(\vl - \ve_u + \ve_v)$. Otherwise, we have $\vl - \ve_u + \ve_v \in \mathcal{V}(\vl)$. Hence $\vl$ satisfies the property:
 $ \displaystyle \forall u,v \in \mathbb{J},\ f(\vl) \geq f(\vl - \ve_u + \ve_v). $
 
 We then apply~\cite[Thm 4.6]{MurotaMConvex} with the M-concave function $f$ on $\vl$, given that $\vl$ is a global maximum of $f$. Therefore $\vl$ is a global maximum of $\pot$ on $\mathcal{D}$.
\end{proof}

We show next that the set of loads corresponding to the different Nash 
equilibria of the game, are all neighbors of each other:


\begin{theorem}
\label{thm:eqLocal}
 Let $\mathcal{E}_\Gamma$ be the set of the loads of the Nash equilibria of the game. Then: 
  $$ \forall \vx,\vy \in \mathcal{E}_\Gamma,\ \vx-\vy = \sum_{u \in \supp^+(\vx-\vy)} \ve_u - \sum_{v \in \supp^+(\vy-\vx)} \ve_v. $$
In other words, all Nash equilibria $\vx,\vy \in \mathcal{E}_\Gamma$ satisfy $\normu{\vec{x}-\vec{y}} \leq 1$.
\end{theorem}
\begin{proof}
 First, if $\mathcal{E}_\Gamma = \{\vx\}$, then the theorem is trivially true.
 Otherwise, assume that there exists $\vx,\vy \in \mathcal{E}_\Gamma$ such that $\normu{\vx-\vy} > 0$. 
Since $\vx$ and $\vy$ are in $\mathcal{D}$, we can write $\vx-\vy$ as
 $ \vx - \vy = \ve_u - \ve_v + \vz$
  for some $u,v \in \mathbb{J}$ and $\vz \in \mathbb{Z}^{J}$ satisfying  $u \neq v$, $\mps{z}{\ve_u} \geq 0$ and $\mps{z}{\ve_v} \leq 0$.
 Therefore, we have $u \in \supp^+(\vx-\vy)$ and $v \in \supp^+(\vy-\vx)$. 
 
 Let $\va = \vx - \ve_u + \ve_v$ and $\vb = \vy - \ve_v + \ve_u$. By Eq.~\eqref{eq:mconvex} we have:
$ f(\vx) + f(\vy) = 2f(\vx) \leq f(\va) + f(\vb). $
 Hence we get $f(\va) = f(\vb) = f(\vx)$ by global maximality of $f(\vx)$. 
Then
$f(\vx)-f(\va)+f(\vy)-f(\vb) = 0 $, 
which in turns implies that $$\frac{N_u}{x_u}+\frac{N_v}{y_v} = \frac{N_u}{y_u+1}+\frac{N_v}{x_v+1}.$$
  Since $x_u \geq y_u+1$ and $y_v \geq x_v+1$, the last equation implies that $x_u = y_u+1$ and $x_v = y_v-1$. Therefore:
  \[ \vx - \vy = \sum_{u \in \supp^+(\vx-\vy)} \ve_u - \sum_{v \in \supp^+(\vy-\vx)} \ve_v. \qedhere\]
\end{proof}

Using this result, we can find a bound over the number of Nash equilibria:

\begin{proposition}
\label{thm:eqBound}
For any setting $\Gamma$, the number of Nash equilibria is upper bounded. More precisely, let $\mathcal{E}_\Gamma$ be the set of the loads of the Nash equilibria of $\Gamma$. Then:
 $$ |\mathcal{E}_\Gamma| \leq \binom{J}{\floor{\frac{J}{2}}}.$$

Further, this bound is tight. Indeed, let $J \geq 2$, $m \in \mathbb{N}^*$ and $\gamma \in \mathbb{R}^+$. We define the game $\Gamma$ by
$ K = \floor{\frac{J}{2}}$ and $\forall j \in \mathbb{J}, N_j = m,\ \gamma_j = \gamma. $ Then $|\mathcal{E}_\Gamma| = \binom{J}{\floor{\frac{J}{2}}}.$
\end{proposition}

The proof of Proposition~\ref{thm:eqBound} is given in Appendix~\ref{app:eqBound}.


Note that the bound is in $\phantom{\Big\{}\hspace{-0.8em}\bigo{\frac{2^{J}}{\sqrt{J}}}$ and that it is independent of the number of seeds $K$. As the number of social media is typically small, then there is a limited number of equilibria.

\section{Algorithmic determination of an equilibrium}

In this section we see how to compute a Nash equilibrium. 
Note that, from Theorem~\ref{thm:eqBound}, computing all the loads of the Nash equilibria would require $\bigom{\frac{2^J}{\sqrt{J}}}$ operations. Further, for each load vector $\vl$ maximizing $\pot$, computing all the corresponding Nash equilibria $\vs \in S$ would require up to $\bigo{K!}$ operations because of the symmetry of the game.

\subsection{Maximization of the potential}
Consider the following optimization mechanism:
\begin{enumerate}[label=Step \arabic*), leftmargin=6em]
 \item Start with some $\vl \in \mathcal{D}$
 \item Find $\vl^*$ the argmax of $\pot$ on $\mathcal{V}(\vl)$
 \item If $\vl^* = \vl$ then stop
 \item Let $\vl = \vl^*$ and repeat from step 2)
\end{enumerate}

Thanks to Theorem~\ref{thm:globMax}, we know that this mechanism converges to a vector of loads of a Nash equilibrium. From the point of view of the seeds, it is similar to a guided best-response mechanism where at each step the seed which could increase the most the potential by changing its strategy is selected. 

The problem is that, in the worst cases, this algorithm visits all the load vectors of the domain $\mathcal{D}$, which leads to $\bigo{K^J}$ steps to find a maximum. However, we can exploit the M-concavity of function $f$ to compute a Nash equilibrium in a far more efficient way. To that end, we adapt the algorithm MODIFIED\_STEEPEST\_DESCENT given in~\cite[p.8]{AlgoMConvex} to our problem, which is presented in Algorithm~\ref{algo:MSD} below.
\begin{algorithm}\DontPrintSemicolon
 \label{algo:MSD}
 \caption{SD\_MAX}
 \KwIn{$\Gamma=(K,(N_j,\gamma_j)_{j\in \mathbb{J}})$}
 \KwOut{A vector in $\mathcal{E}_\Gamma$}
 Let $\vl = K\ve_1$ and $\vb = 0 \in \mathbb{Z}^J$\;
 \While{$\exists u,\ \ell_u - 1 \geq b_u$}
 {
   Compute $v \in \argmax_{t \in \mathbb{J}} \left( \frac{N_t}{\ell_t+1}-\gamma_t \right)$\;
   $b_v \leftarrow \ell_v+1$\;
   $\vl \leftarrow \vl - \ve_u + \ve_v$ \; 
 }
 \Return $\vl$
\end{algorithm}

\begin{proposition}
 Algorithm~\ref{algo:MSD} terminates, returns a vector in $\mathcal{E}_\Gamma$ with a time complexity in $\bigo{KJ^2}$.
\end{proposition}
\begin{proof}
 We implemented the active domain $B$ of the algorithm used in~\cite[p.8]{AlgoMConvex} by a vector $\vb$ satisfying:
 $$ B = \left\{ \vx\ |\ \sum_{u \in \mathbb{J}} x_u = K \text{ and } \forall u\in \mathbb{J},\ x_u \geq b_u \right\}.$$
 
 We then remarked that we do not need to compute the potential since $v= \argmax_{t \in \mathbb{J}}f(\vl-\ve_u+\ve_t)$ is equivalent to:
  \begin{align*}
   & \forall t \in \mathbb{J},\ f(\vl-\ve_u+\ve_v) - f(\vl-\ve_u+\ve_t) \geq 0 \\
   &\Leftrightarrow \forall t \in \mathbb{J},\ \frac{N_v}{\ell_v+1}-\gamma_v - \frac{N_t}{\ell_t+1}+\gamma_t \geq 0 \\
   &\Leftrightarrow v = \argmax_{t \in \mathbb{J}}\left(\frac{N_t}{\ell_t+1}-\gamma_t\right). 
  \end{align*}

 We can then apply the same analysis as in~\cite{AlgoMConvex} for the correctness of the algorithm.

 The quantity $0 \leq \sum_{u \in \mathbb{J}} (K - b_u) \leq KJ$ decreases by at least one at each step of the algorithm. Therefore, the algorithm terminates with at most $KJ$ iterations. Moreover, finding a $u$ satisfying the loop condition and computing the value of $v$ can be done in $\bigo{J}$. Hence, this algorithms has a time complexity in $\bigo{KJ^2}$. 
\end{proof}
\iffullarticle
\else
We have also designed a refinement of this algorithm  
thanks to some scaling properties of function $f$, 
which gives a complexity in
$\bigo{J^3\log K/J}$. It is omitted due to length requirements.
\fi
\subsection{An efficient learning mechanism}
Note that the previous algorithm starts with some arbitrary load vector in $\mathcal{D}$ and then iteratively finds the best improvement until reaching a global maximum. 

Instead, we propose a novel approach in which the seeds arrive one by one. We then show that at each arrival of seed $k$, the strategy $s_k$ can be computed in such way that after all arrivals, the resulting vector $\vs$ is a Nash equilibrium. This approach relies on the following theorem:

\begin{theorem}
 \label{thm:eqIncK}
 Let $\Gamma=(K,(N_j,\gamma_j)_{j\in \mathbb{J}})$ be a setting of the game and $\Gamma'=(K+1,(N_j,\gamma_j)_{j\in \mathbb{J}})$ be the setting obtained by adding an extra seed on $\Gamma$. Let $\vs$ be a Nash equilibrium for $\Gamma$ and $\vec{\sigma}$ the strategy profile of $\Gamma'$ in which the $K$ first seeds choose the same strategy as in $\vs$ (i.e. $s_k = \sigma_k$ for all $k\leq K$) and the additional seed chooses one of the social media which maximizes its payoff. Then $\vec{\sigma}$ is a Nash equilibrium for $\Gamma'$.
 
 Formally, let $\vl \in \mathcal{E}_\Gamma$ be the load of some Nash equilibrium of $\Gamma$ and $w\in \mathbb{J}$. Then
 $$ \vl+\ve_w \in \mathcal{E}_{\Gamma'} \Leftrightarrow w \in \argmax_{t \in \mathbb{J}} \left( \frac{N_t}{\ell_t+1} - \gamma_t \right).$$
\end{theorem}
\begin{proof}
 Let 
$  w \in \argmax_{t \in \mathbb{J}} \left( \frac{N_t}{\ell_t+1} - \gamma_t \right).$
 We proceed to show that $\vl + \ve_w$ is in $\mathcal{E}_{\Gamma'}$.
 
 Let $u,v \in \mathbb{J}$ such that $\ell_u+\delta(u,w) > 0$ and $u \neq v$. We need to show that $\pot(\vl+\ve_w) \geq \pot(\vl+\ve_w-\ve_u+\ve_v)$. There are three cases detailed below.
 
 First, consider that $u \neq w$ and $v \neq w$. Then $\pot(\vl+\ve_w) - \pot(\vl+\ve_w-\ve_u+\ve_v) = \pot(\vl)-\pot(\vl-\ve_u+\ve_v)$ so it is proven in this case.
 
 Second, consider that $v = w$. We have
 \begin{align*}
  & \pot(\vl+\ve_w) - \pot(\vl+\ve_w-\ve_u+\ve_v)\\
  &= \pot(\vl+\ve_w)-\pot(\vl+2\ve_w-\ve_u)\\
  &=\frac{N_u}{\ell_u}-\frac{N_w}{\ell_w+2}+\gamma_w-\gamma_u\\
  &\geq \frac{N_u}{\ell_u}-\frac{N_w}{\ell_w+1}+\gamma_w-\gamma_u \geq 0
 \end{align*}
 since $\vl \in \mathcal{E}_\Gamma$.
 
 Third, consider that $u = w$. We have
 \begin{align*}
  &\pot(\vl+\ve_w) - \pot(\vl+\ve_w-\ve_u+\ve_v)\\
  &= \pot(\vl+\ve_w) - \pot(\vl+\ve_v)\\
  &= \frac{N_w}{\ell_w+1}-\frac{N_v}{\ell_v+1}+\gamma_w-\gamma_v \geq 0
 \end{align*}
 by definition of $w$.
 
 The reciprocal follows from the last formula: if $\frac{N_w}{\ell_w+1}-\gamma_w$ was not maximal, then there would be some $v$ such that $\pot(\vl+\ve_w) < \pot(\vl+\ve_v)$. Hence $\vl + \ve_w$ would not be in $\mathcal{E}_{\Gamma'}$ from Theorem~\ref{thm:globMax}.
\end{proof}

We use Theorem~\ref{thm:eqIncK} to build an efficient algorithm finding a vector of loads of a Nash equilibrium  (Algorithm~\ref{algo:OrderLearning}): it begins with no seed and simulates $K$ times the arrival of a seed maximizing its payoff.
\setlength{\textfloatsep}{-1pt}
\begin{algorithm}\DontPrintSemicolon
 \label{algo:OrderLearning}
 \caption{ORDER\_LEARNING}
 \KwIn{$\Gamma=(K,(N_j,\gamma_j)_{j\in \mathbb{J}})$}
 \KwOut{A vector in $\mathcal{E}_\Gamma$}
 Let $\vl = \vec{0} \in \mathbb{Z}^J$ and $k = 1$\;
 \While{$k \leq K$ }
 {
   Compute $w \in \argmax_{t \in \mathbb{J}} \left( \frac{N_t}{\ell_t+1}-\gamma_t \right)$\;
   $\vl \leftarrow \vl + \ve_w$\;
   $k \leftarrow k + 1$\; 
 }
 \Return $\vl$
\end{algorithm}
\setlength{\textfloatsep}{12pt}
\begin{proposition}
 Algorithm~\ref{algo:OrderLearning} terminates and returns a vector in $\mathcal{E}_\Gamma$ with a time complexity in $\bigo{KJ}$.
\end{proposition}
\begin{proof}
 Let $\Gamma_k=(k,(N_j,\gamma_j)_{j \in \mathbb{J}})$ for $k \in \{0,...,K\}$. Since $\vec{0} \in \mathbb{Z}^J$ is a vector of loads of a Nash equilibrium of $\Gamma_0$, then, from Theorem~\ref{thm:eqIncK}, at the end of the $k$th iteration of the loop, $\vl$ is a vector of loads of a Nash equilibrium of $\Gamma_k$, hence the correctness of Algorithm~\ref{algo:OrderLearning}.
 
 Since we can compute $w$ in $\bigo{J}$ and there are $K$ iterations of the loop, then the time complexity of Algorithm~\ref{algo:OrderLearning} is in $\bigo{KJ}$.
\end{proof}


\section{Asymptotic Behavior}
In this section, we discuss the form of the Nash equilibria when we have a lot more seeds than social media, so $K \gg J$. We are interested in this case in practice as the activity in the Internet tend to be concentrated in a restricted number of famous websites. 

\subsection{Intuition}

First, we can make a hypothesis about the asymptotic behavior of this game when $K \to \infty$ according to the form of the potential. Recall that
$ H_n \underset{n \to \infty}{\sim} \ln(n) + \mu $,
where $\mu$ is the Euler-Mascheroni constant. Then in order to find approximate Nash equilibria, we can study the function
$$P(\vl) = \sum_{j,\ \ell_j>0} \left( N_j \ln (\ell_j) - \gamma_j \ell_j \right).$$
We can see that, for large values of $\ell_j$, the linear part in $\gamma_j \ell_j$ is determinant compared to the logarithm part in $N_j \ln(\ell_j)$. Therefore we can make the hypothesis that when the quantity $\sum_j \ell_j = K$ is large enough, then the only $\ell_j$ that continue to increase are the ones with minimal cost. Then, it seems natural that all social media with minimal cost would behave as if they were in a subgame where new seeds would only choose them.

\subsection{Asymptotic Analysis}
Following our hypothesis, we define $\gm$ the minimal cost and $G$ the set of social media with minimal costs:
$$ \gm = \min_j \gamma_j \text{ and } G = \argmin_j \gamma_j. $$

We know, thanks to Theorem~\ref{thm:eqIncK}, that when $K$ increases, the coordinates of the loads of the Nash equilibrium we consider can only increase. We proceed to show our intuition. In the following, we note $\Gamma_K = (K,(N_j,\gamma_j)_{j \in \mathbb{J}})$ and $\mathcal{E}_K = \mathcal{E}_{\Gamma_K}$. We study the vectors in $\mathcal{E}_K$ obtained with the mechanism implemented in Algorithm~\ref{algo:OrderLearning}. Let $\vlk{K}$ be the vector in $\mathcal{E}_K$ obtained after the $K$th iteration of the loop in the algorithm.

\begin{theorem}
 \label{thm:asymptoticForm}
 When $K$ goes to infinity, at the Nash equilibria, the social media are divided into two groups:
 \begin{itemize}
  \item The loads of the social media with non-minimal cost stop increasing when they reach a constant. Formally:
  $$\forall j \in \mathbb{J}\setminus G,\ \ellk{K}_j \underset{K \to \infty}{\longrightarrow} \ceil{\frac{N_j}{\gamma_j-\gm}}-1.$$
  \item The loads of the social media with minimal cost goes to infinity, and the proportion of seeds a social medium get among the one with minimal cost is equal to its market share. Formally: 
 \end{itemize}
    $$\forall w \in G,\ \cfrac{\ellk{K}_w}{\sum_{t \in G} \ellk{K}_t} \underset{K \to \infty}{\longrightarrow} \cfrac{N_w}{\sum_{t \in G} N_t}.$$
\end{theorem}

The proof of Theorem~\ref{thm:asymptoticForm} is given in Appendix~\ref{app:asymptBehavior}.

\section{Numerical Results}

Figure~\ref{fig:asymptoticBehavior} shows the convergence of the equilibrium when the number of seeds $K$ grows large. Note that there may be up to $3$ equilibria and that the plots of the figure correspond to the outputs of Algorithm~\ref{algo:OrderLearning}. The asymptotes obtained in Theorem~\ref{thm:asymptoticForm} are represented in dashed lines with colors matching those of the loads of their associated social media (SM). The SM $2$ and $3$, which have minimal cost, have loads growing to infinity with the number of seeds. The asymptote of SM $2$ has a higher slope than that of $3$ because it has a higher number of subscribers ($N_2>N_3$).
Finally, while for large values of seeds the cost of the social media is predominant, in contrast, for low values of seeds, the number of customers $M$ plays the larger role in determining the loads of the different social media.

\setlength{\abovecaptionskip}{4pt plus 2pt minus 2pt}
\begin{figure}[htb]
\resizebox{\columnwidth}{!}{\input{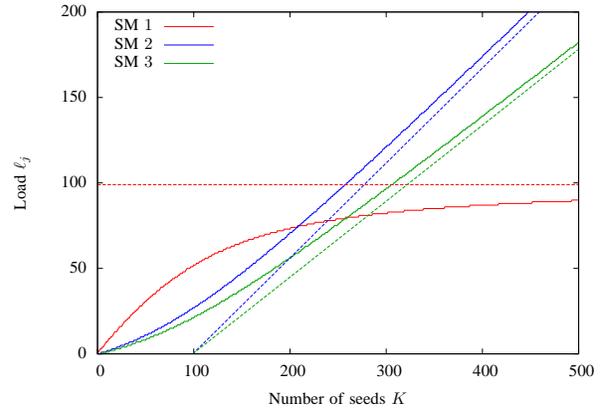}}
\caption{Convergence to the asymptotic behavior. Case with: $(N_1,\gamma_1) = (100,2)$, $(N_2,\gamma_2) = (25,1)$ and $(N_3,\gamma_3) = (20,1)$.}
\label{fig:asymptoticBehavior}
\end{figure}
\begin{figure}[htb]
\resizebox{\columnwidth}{!}{\input{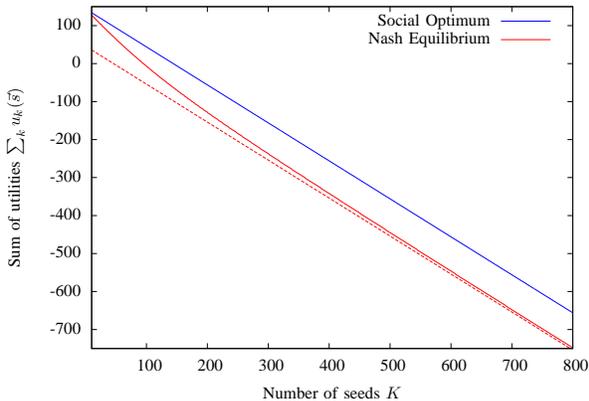}}
\caption{Sum of utilities at equilibrium compared with the social optimum (same setting that Figure~\ref{fig:asymptoticBehavior}).}
\label{fig:PoA}
\end{figure}

Figure~\ref{fig:PoA} shows the evolution of the social welfare, that is, the sum of total utilities, $\sum_k u_k(\vs)$, at the Nash equilibrium and at the social optimum. The asymptotic behavior at the Nash equilibria is given by $-\gamma_m K$ from Theorem~\ref{thm:asymptoticForm}.
Further, let ${\cal L} = \{j, N_j \geq \gamma_j\}$. Recall that the social optimum is the strategy vector maximizing the social welfare. Then, for $K$ large enough, the social optimum satisfies $ \max_{\vec{\ell}} \sum_{j, \ell_j>0} (N_j-\gamma_j \ell_j) = \sum_{j \in \cal{L}} (N_j-\gamma_j)-\gamma_m (K-|{\cal L}|) \sim - \gamma_m K$.  Hence, as the number of seeds grows to infinity, the price of anarchy converges to $1$. 

Finally, Figures~\ref{fig:NSensitivity} and \ref{fig:GSensitivity} show the sensitivity of the equilibria with respect to $N$ and $\gamma$ for a case with $J=2$ social media. 

We observe that the load of a social medium is increasing with its number of customers, as expected (Fig. \ref{fig:NSensitivity}). Further, if the dissemination cost of SM $2$ is higher or equal to that of SM $1$ and if it has no customer, then its load is zero, as exhibited in the red and blue plots. Otherwise, even though it has no customer, if its cost is minimal, it will receive some seeds (green plot). Finally, note that as the cost of SM $1$ decreases, the number of customers in the SM $2$ has lower effect of the evolution of the load $\ell_2$.

We also observe that the load of a social medium is decreasing with its dissemination cost, as expected (Fig.~\ref{fig:GSensitivity}). 
Further, numerical results show that the load decreases more abruptly for lower number of users $N_1$, but that the drop occurs for larger values of $\gamma_2$.

\begin{figure}[htb]
\resizebox{\columnwidth}{!}{\input{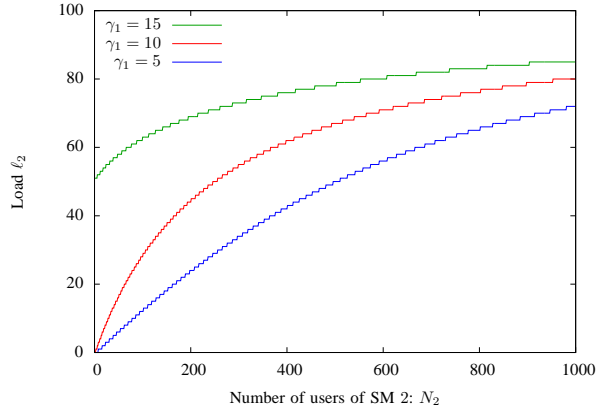}}
\caption{Influence of the number of users. Case with 2 social media: $N_1=250$, $\gamma_2=10$ and $K=100$.}
\label{fig:NSensitivity}
\end{figure}

\begin{figure}[htb]
\resizebox{\columnwidth}{!}{\input{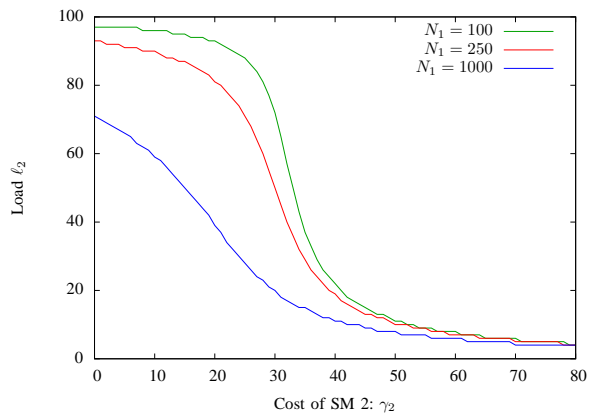}}
\caption{Influence of the dissemination cost. Case with 2 social media: $\gamma_1 = 30$, $N_2=250$ and $K=100$.}
\label{fig:GSensitivity}
\end{figure}

\section{Conclusion}
In this paper, we have studied competition for popularity of seeds among several social media. We have shown that the game is equivalent to a congestion game
and hence  has a potential. We then studied the properties of the potential in terms of $M$-concavity. We have shown that there may exist several Nash equilibria, all belonging to a single neighborhood and provided examples where the number of equilibria is maximal. We have provided a novel efficient learning algorithm based on a remarkable property of the Nash equilibria in some subgames. We also investigated the asymptotic behavior of the equilibria of the game and the price of anarchy. 
As future work, we will study the underlying competition among the social media in the Stackelberg setting for a discrete number of seeds: according to their number of subscribers (who consume content), how could they appropriately set up their prices? We further plan to extend our model to the case where seeds have different dissemination utilities for sending to the various media.
The game is no more equivalent to a congestion game but turns out to 
be equivalent to crowding games \cite{milch}. This allows to show existence of
(pure) equilibria but best response policies need not converge, as
there need not be a potential anymore. Thus, designing learning algorithms
for this extension
is yet an open problem.

\bibliographystyle{IEEEtran}
\bibliography{discrete}
\appendices
\section{Proof of Proposition~\ref{thm:eqBound}}
\label{app:eqBound}

\subsection{Proof of the upper bound $ |\mathcal{E}_\Gamma| \leq \binom{J}{\floor{\frac{J}{2}}}.$}
\begin{lemma}
 \label{lem:eqRegul}
 Let $\vx$ be a load vector at a Nash equilibrium, $\vx \in \mathcal{E}_\Gamma$, and $u$ a social medium, $u \in \mathbb{J}$. Then
 $$ \vspace{-0.5em}\left( \exists \vy \in \mathcal{E}_\Gamma,\ x_u > y_u \right) \Rightarrow \left( \forall \vz \in \mathcal{E}_\Gamma,\ x_u \geq z_u \right). $$
\end{lemma}
\begin{proof}
 Assume that there exists $\vy$ and $\vz$ in $\mathcal{E}_\Gamma$ such that $x_u > y_u$ and $x_u < z_u$.
 Then, by Theorem~\ref{thm:eqLocal}, $y_u = x_u - 1$ and $z_u = x_u + 1$. Hence $z_u - y_u = 2$ which contradicts Theorem~\ref{thm:eqLocal}.
\end{proof}

\begin{lemma}
 \label{lem:lbinom}
 Let $\alpha \in \mathbb{J}$. Then
 $$ \sum_{k=0}^{\min(\alpha,J-\alpha)} \binom{\alpha}{k} \binom{J-\alpha}{k} = \binom{J}{\alpha}.$$
\end{lemma}
\begin{proof}
 We show this result using a combinatorial argument. First, note that since $ \displaystyle \binom{J}{\alpha} = \binom{J}{J-\alpha}$, then one can restrict the analysis to the case where
$\alpha \leq J-\alpha$. 

We want to select $\alpha$ elements in $\mathbb{J}$. To do that, we partition the set $\mathbb{J}$ into two subset $A$ and $B$ such that $|A| = \alpha$ (so $|B|=J-\alpha$). 

Selecting $\alpha$ elements in $\mathbb{J}$ amounts to choosing $k$ the number of elements we select in $B$, then select these $k$ elements and finally select $\alpha-k$ elements in $A$. Therefore
 \[ 
\begin{array}{l@{\,}l}
\displaystyle \binom{J}{\alpha} &= \displaystyle \sum_{k=0}^{|A|} \binom{|B|}{k} \binom{|A|}{\alpha-k} = \sum_{k=0}^{\alpha} \binom{\alpha}{\alpha-k} \binom{J-\alpha}{k}\\[2em]
& \displaystyle = \sum_{k=0}^{\min(\alpha,J-\alpha)} \binom{\alpha}{k} \binom{J-\alpha}{k}.\hfill \qedhere
\end{array} \]
 
\end{proof}

We can now proceed to the proof of Theorem~\ref{thm:eqBound}:
\begin{proof}
 Let $\vx \in \mathcal{E}_\Gamma$ and:
 $$ \begin{array}{l@{\,}l}
\mathcal{U} &= \left\{ u\, |\, \exists \vy \in \mathcal{E}_\Gamma,\, x_u > y_u \right\} = \hspace{-3pt}\bigcup\limits_{\vy \in \mathcal{E}_\Gamma}\hspace{-3pt} \supp^+(\vx-\vy) \text{ and }\\
\mathcal{V} &= \bigcup\limits_{\vy \in \mathcal{E}_\Gamma} \supp^+(\vy-\vx).
\end{array}
$$ 
 
 By Lemma~\ref{lem:eqRegul}, we have $\mathcal{U} \cap \mathcal{V} = \emptyset$ and $|\mathcal{U}|+|\mathcal{V}| \leq J$. We then define the set 
$\mathcal{A}=$
$$\{\vx\} \, \cup 
\hspace{-1em} \bigcup_{k=1}^{\min(|\mathcal{U}|,|\mathcal{V}|)} \hspace{-3pt} \left\{ 
\left.
\vx - \sum_{u \in C} \vec{e_u} + \sum_{v \in D} \vec{e_v} \, \right|
\begin{array}{@{\,}l@{\,}}
C \subset \mathcal{U},\\ 
D \subset \mathcal{V},\\ |C|=|D|=k \end{array} 
\right\}.$$
 
 We know by Proposition~\ref{thm:eqBound} and Lemma~\ref{lem:eqRegul} that all vectors in $\mathcal{E}_\Gamma$ are of the form given in the previous expression, hence $\mathcal{E}_\Gamma \subset \mathcal{A}$. 
 
 Let $\alpha = |\mathcal{U}|$. We have $|\mathcal{V}| \leq J - \alpha$. Then
 \begin{align*}
  |\mathcal{E}_\Gamma| &\leq |\mathcal{A}| = 1 + \sum_{k=1}^{\min(|\mathcal{U}|,|\mathcal{V}|)} \binom{|\mathcal{U}|}{k}\binom{|\mathcal{V}|}{k}\\
		&\leq 1 + \sum_{k=1}^{\min(\alpha,J-\alpha)} \binom{\alpha}{k} \binom{J-\alpha}{k}. 		
 \end{align*}
 
 We conclude the proof by applying Lemma~\ref{lem:lbinom}, using the increasing property of function $\binom{J}{.}$ over $\{0,...,\lfloor J/2 \rfloor \}$ and the fact that $\binom{J}{p} = \binom{J}{J-p}$ for all $p$: 
 
 \[ |\mathcal{E}_\Gamma| \leq \binom{J}{\alpha} \leq \binom{J}{\floor{\frac{J}{2}}}. \qedhere \]
\end{proof}

\subsection{A Tight Class of Settings}
Let $J \geq 2$, $m \in \mathbb{N}^*$ and $\gamma \in \mathbb{R}^+$. We define the game $\Gamma$ by
$ K = \floor{\frac{J}{2}} \text{ and } \forall j \in \mathbb{J}, N_j = m,\ \gamma_j = \gamma$.
\begin{lemma}
 \label{lem:exEq1}
 The Nash equilibria of game $\Gamma$ satisfy the property:
 $$ \vl \in \mathcal{E}_\Gamma \Rightarrow \exists A \subset \mathbb{J},\ |A| = \floor{\frac{J}{2}} \text{ and } \vl = \sum_{u \in A} \ve_u. $$
\end{lemma}

\begin{proof}
 Assume that there exists $\vec{x} \in \mathcal{E}_\Gamma$ and $u \in \mathbb{J}$ such that $x_u > 1$. 
 Since $K < J$, there exists $v \in \mathbb{J}$ such that $x_v = 0$. Consider the vector
 $ \vy = \vx - x_u \ve_u + x_u \ve_v. $
 
 Since all the $N_j$ and $\gamma_j$ are equal, the potential of $\vy$ is equal to the potential of $\vx$. Therefore $\vy \in \mathcal{E}_\Gamma$. But we have $y_v - x_v = x_u > 1$ which contradicts Theorem~\ref{thm:eqLocal} and concludes the proof. 
\end{proof}

Since $\mathcal{E}_\Gamma \neq \emptyset$, let $\vx \in \mathcal{E}_\Gamma$. By Lemma~\ref{lem:exEq1}, we can note $\vx = \sum_{u \in A} \ve_u$ for some $A \subset \mathbb{J}$. 
Let $B \subset \mathbb{J}$ verifying $|B| = \floor{\frac{J}{2}}$ and $\vy = \sum_{v \in B} \ve_v$. Then, we have
\begin{align*} 
  \pot(\vx) &= \sum_{u \in A} (m - \gamma) = \floor{\frac{J}{2}} (m - \gamma) \\
            &= \sum_{v \in B} (m - \gamma) = \pot(\vy). 
\end{align*}
$$ \text{Therefore, } |\mathcal{E}_\Gamma| = \left| \left\{A \subset \mathbb{J}\ |\ |A| = \floor{\frac{J}{2}}\right\} \right| = \binom{J}{\floor{\frac{J}{2}}}.$$

\section{Proof of the asymptotic behavior}
\label{app:asymptBehavior}
Note that, by definition of the learning mechanism implemented in Algorithm~\ref{algo:OrderLearning}, for all $K \in \mathbb{N}$ and $j \in \mathbb{J}$ we have
\begin{equation}
 \label{eq:learnMech}
 \ellk{K+1}_j \hspace{-2pt}= \ellk{K}_j+1 \Rightarrow j \in \argmax_{t \in \mathbb{J}} \left( \frac{N_t}{\ell_t+1} - \gamma_t \right).
\end{equation}

\subsection{Social media with non minimal cost}
We want to prove that
\begin{equation}
 \label{eq:asBlocked}
 \forall j \in \mathbb{J}\setminus G,\ \ellk{K}_j \underset{K \to \infty}{\longrightarrow} \ceil{\frac{N_j}{\gamma_j-\gm}}-1.
\end{equation}

We begin by proving the following two lemmas.

\begin{lemma}
\label{lem:maxLimit}
 The quantity $$\mk{K} = \max_{t \in \mathbb{J}} \left( \frac{N_t}{\ellk{K}_t+1}-\gamma_t \right)$$ 
 is arbitrarily close to $\gm$ for $K$ large enough\footnote{We denote by $K$ ``large enough'' the fact that there exists some $K_0$ such that the property is verified for all $K>K_0$.}. 
\end{lemma}
\begin{proof}
 First, this quantity is decreasing. Moreover, by definition of the $\vlk{K}$, we have that for all $K$, 
 $\displaystyle \sum_j \ellk{K}_j = K \underset{K \to \infty}{\longrightarrow} \infty. $ 
 Therefore, there exists some $u \in \mathbb{J}$ such that $\ellk{K}_u \underset{K \to \infty}{\longrightarrow} \infty$. It means that there exists $(K_n)_{n\in \mathbb{N}}$ such that 
 $ \forall n,\ \ellk{K_n+1}_u = \ellk{K_n}_u + 1$,  
 which implies by~\eqref{eq:learnMech} that
$\forall n,\ \frac{N_u}{\ellk{K_n}_u+1}-\gamma_u = \mk{K_n}.$
 
 Hence $\mk{K_n}$ is arbitrarily close to $-\gamma_u$ for $n$ large enough. We conclude by noticing that $-\gamma_u\leq-\gm$. 
\end{proof}

\begin{lemma}
 \label{lem:eqInfty}
 Let $K>0$ and $u \in \mathbb{J}$. Then 
 \begin{align*}
  \frac{N_u}{\ellk{K}_u+1} - \gamma_u > -\gm 
  \Leftrightarrow \exists\, K'\!>\!K,\ \ellk{K'}_u > \ellk{K}_u.
 \end{align*}
\end{lemma}

\begin{proof}
 First, assume that $\frac{N_u}{\ellk{K}_u+1} - \gamma_u \leq -\gm$. Then for some $w \in G$ and for all $K'\geq K$ we have
 $$ \frac{N_u}{\ellk{K}_u+1} - \gamma_u < \frac{N_w}{\ellk{K'}_w+1} - \gamma_w $$
 since $\gamma_w = \gm$. This implies that $\frac{N_u}{\ellk{K}_u+1} - \gamma_u < \mk{K'}$. Therefore, for all $K'>K$,~\eqref{eq:learnMech} leads to
 $\ellk{K'}_u = \ellk{K}_u$.
 
 Then assume that $\frac{N_u}{\ellk{K}_u+1} - \gamma_u > -\gm$. According to Lemma~\ref{lem:eqInfty}, $\mk{K'}$ is arbitrarily close to $-\gm$ for $K'$ large enough. Therefore there exists $K'>K$ such that 
 $\displaystyle \mk{K'} < \frac{N_u}{\ellk{K}_u+1} - \gamma_u.$
 Hence $\ellk{K'}_u > \ellk{K}_u$ which concludes the proof.
\end{proof}

We can now proceed with the proof of~\eqref{eq:asBlocked}.

Let $j \in \mathbb{J} \setminus G$. We know by Lemma~\ref{lem:eqInfty} that $\ellk{K}_j$ increases with $K$ as long as $\displaystyle \frac{N_j}{\ellk{K}_j+1} - \gamma_j > -\gm. $
 
 Therefore, for $K$ large enough we have
 $$\ellk{K}_j = 1 + \max \left\{ p \in \mathbb{N}\ |\ \frac{N_j}{p+1}-\gamma_j > -\gm\right\}.$$
 
 Then, let $p\in \mathbb{N}$. Since $j \in \mathbb{J}\setminus G$, we have $\gamma_j>\gm$. We solve
  $\displaystyle \frac{N_j}{p+1}-\gamma_j > -\gm \Leftrightarrow p+1 < \frac{N_j}{\gamma_j-\gm}.$
 Hence
 $ \displaystyle \ellk{K}_j + 1 = \ceil{\frac{N_j}{\gamma_j-\gm}} $
 which concludes the proof.

\subsection{Social media with minimal cost}
We can directly conclude from Lemma~\ref{lem:eqInfty} that the load of any social medium having a minimal cost goes to infinity as $K$ increases. Formally:
\begin{equation}
\label{eq:minCostInfty}
 \forall w \in G,\ \ellk{K}_w \underset{K \to \infty}{\longrightarrow} \infty.
\end{equation}

Now we proceed to find the values of $\ellk{K}_w$ for the social media with minimal cost. Let $K$ be large enough so that ~\eqref{eq:asBlocked} is verified. Let 
$ \displaystyle K_G = K - \sum_{j \in \mathbb{J} \setminus G} \ellk{K}_j$
be the number of seeds sharing the social media in $G$ and
$$\mathcal{D}_G = \left\{ (x_t)_{t\in G}\ | \sum_{t \in G} x_t = K_G \text{ and } \forall t\in G, x_t > 0 \right\}.$$

Consider the game $\Gamma_G = (K_G,(N_t,\gm)_{t \in G})$. From~\eqref{eq:minCostInfty}, the loads of the social media in $G$ can be arbitrarily high with $K$ large enough, so we determine an approximation of a load of a Nash equilibrium for the social media in $G$ by solving
$$ \max_{\vx \in \mathbb{R}^G} P(\vx)=\sum_{t \in G} \left(N_t \ln(x_t) - \gm x_t\right) \text{ s.t. } \vx \in \mathcal{D}_G.$$

Since $P$ is concave, we apply a Lagrangian maximization method.
Let $L$ be the Lagrangian for this problem:
$$ L(\vx,\lambda)\! =\! P(\vx) - \lambda \left( \sum_{t \in G} x_t - K_G \right), $$
where $\lambda$ and the $x_t$ are nonnegative. 

Since $P$ is concave, the unique maximum $\vxe$ verifies
$ \displaystyle \forall t\in G,\ \frac{\partial L}{\partial x_t}(\vxe) = 0.$
Therefore, we get that for any $t$:\\
$ \displaystyle \frac{N_t}{x^*_t} - \gm - \lambda = 0 \Leftrightarrow x^*_t = \frac{N_t}{\gm + \lambda}. $

\vspace{0.4em}Now we determine the value of $\lambda$:
\begin{align*}
 \sum_{t\in G} x^*_t = K_G &\Rightarrow \sum_{t\in G} \frac{N_t}{\gm + \lambda} = K_G \\
		&\Rightarrow \lambda = \frac{1}{K_G} \left(\sum_{t\in G} N_t\right) - \gm.
\end{align*}

Hence
$ \forall w\in G,\ x^*_w = K_G\frac{N_w}{\sum_{t\in G} N_t}. $

Thanks to~\eqref{eq:minCostInfty} and since $H_n-\mu \underset{n \to \infty}{\sim} \ln n$, we finally get that
$ \displaystyle \forall w\in G,\ \frac{\ellk{K}_w}{\sum_{t \in G} \ellk{K}_t} \underset{K \to \infty}{\longrightarrow} \frac{N_w}{\sum_{t \in G} N_t}.$


\end{document}